\documentclass[12pt, onecolumn,draftcls]{IEEEtran}

\usepackage{mhchem}
\usepackage{enumerate}
\usepackage{amsmath,amssymb,bm,amsfonts}
\usepackage{graphics,graphicx,epsfig,subfigure}
\usepackage{algorithm,algorithmic}
\usepackage{theorem}
\usepackage{threeparttable}
\usepackage{stfloats}
\usepackage{cite}
\usepackage{multirow}
\usepackage{mathrsfs}
\usepackage{amsmath,bm,bbm}
\usepackage{color}
\usepackage{multirow}
\usepackage{subeqnarray}
\usepackage{cases}
\usepackage{colortbl}
\usepackage{wasysym}
\usepackage[dvipsnames]{xcolor}
\usepackage{url}
\newtheorem{proof}{Proof}

\newtheorem{mylemma}{Lemma}

\newtheorem{myproperty}{Property}
\begin{document}
\newcommand*{\QEDA}{\hfill\ensuremath{\blacksquare}}
\def\arrow{{\rightarrow}}
\def\E{{\mathcal{E}}}
\def\N{{\mathcal{N}}}
\def\B{{\mathcal{B}}}
\def\G{{\mathcal{G}}}
\def\I{{\mathcal{I}}}
\def\diag{{\textrm{diag}}}
\def\i {{ -i}}

\title{Greening Geographical Power Allocation for Cellular
Networks}
\author{Yanju Gu}
\maketitle

\IEEEpeerreviewmaketitle
\begin{abstract}
Harvesting energy from nature (solar, wind etc.)  is envisioned as a key enabler for realizing green wireless networks.
However, green energy sources are geographically distributed and the power amount is random which may not enough to power a base station by a single energy site.
Burning brown energy sources such as coal and crude oil, though companied with carbon dioxide emission, provides stable power.
 In this paper, without sacrificing communication quality, we investigate how to perform green energy allocation to abate the dependence on brown energy with hybrid  brown  and green energy injected in power networks.
We present a comprehensive framework
to characterize the performance of hybrid green and brown energy empowered cellular network.
Novel performance  metric ``bits/ton\ce{CO2}/Hz''  is proposed to evaluate the greenness of the communication network.
As  green energy is usually generated from distributed geographical locations and is time varying, online geographical  power allocation algorithm is proposed to maximize the greenness of communication network considering electricity transmission's physical laws i.e., Ohm's law and Kirchhoff's circuit laws.
Simulations show that
geographically distributed green energy sources complement each other by improving the communication capacity while saving brown energy consumption.
Besides, the penetration of green energy can also help reduce power loss on the transmission breaches.
\end{abstract}

%
%
\section{Introduction}
\subsection {Context and Motivation}
Telecommunications
industry consumes $2\%$ of the total electricity  generation worldwide and
base stations energy consumption taking account more than $80\%$ \cite{Bhargava}.
Energy reduction in base stations has been studied  in many ways: such as hardware design (e.g., more
energy efficient power amplifiers \cite{HuaweiBS} and
topological management (e.g., the deployment
of relays and/or micro BSs \cite{depoly,du2013network}).
Many base station equipment manufacturers have begun to offer a number of cost friendly solutions to  reduce power demands of base stations.
For example, Nokia Networks launches more than $20$ products and services on its energy efficient base stations at Mobile World Congress $2015$.

For the cellular networks with multiple base stations, traffic-driven base station switching-on/off strategy has been proposed \cite{Traffic-driven} to turn off a base station one by one that will minimally affect the network.
Similar idea has also been further  studied in \cite{on-off, clockconf} with distributed control algorithm design.
Instead of putting BSs into sleep, \cite{SoftOn-OFF} tactically reduces the coverage (and the power usage) of each BS, and strategically places microcells (relay stations) to offload the traffic transmitted to/from base stations in order to save total power consumption.

In stead of simply saving  base stations consumption, other efforts have been put in  empowering celluar networks with \emph{green energy} such as sunlight and wind energy.
In contrast,
power source  emits \ce{CO2} while burned has negative impact on environment is referred as \emph{brown energy}.

The  fossil fueled  power plants
have great
negative impacts on the environment as they
emit a large part of man-made greenhouse gas  to the atmosphere.
It has been estimated \cite{greennetwork} that cellular networks empowered by brown energy will
consume so much energy that their  \ce{CO2} emissions equivalent
will be $0.4$ percent of the global value by $2020$.
\ce{CO2} is a heat-trapping ``greenhouse" gas which represents a negative externality on the climate system.
Generation of electricity relies
heavily on brown energy.
Countries such as Australia, China, India, Poland and South Africa produce over two-thirds of their electricity and heat through the combustion of coal.
To control the use of brown energy, \ce{CO2}
has already been priced and thereby giving a financial value to each tonne of emissions saved.
A sufficiently high carbon price also promotes investment in clean, low-carbon technologies.
In France, the new carbon tax is 25 US dollars per tonne of \ce{CO2} and
In Switzerland US the price is even $34.20$ per  tonne \ce{CO2}
The total estimated income from the carbon tax would have been between $4.5$ billion euros annually.

In contrast, non-combustion energy sources such as sunlight and wind energy not convert hydrocarbons to \ce{CO2}.
Once the system is paid for, the owner of a renewable energy system will be producing their own renewable electricity for essentially no cost and can sell the excess to the local utility at a profit.
Therefore, green cellar network benefits both environment protection and economy cost.
To build a green cellular network, we need to first improve base station which dominates the  contribution factor to overall energy consumption.

However, green energy has special properties different from  traditional brown energy, and  how to
efficiently stylize green energy is challenging.
As green energy harvested is usually random, most of existing works focus on
 designing new transmission
strategies that can best take advantage of and adapt to the
random energy arrivals \cite{OzelJsac}.
Optimal packet scheduling is proposed in \cite{Scheduling} to
adaptively change the transmission rate according to the traffic load and available energy, such that the time by which all packets are delivered is minimized.

However, green energy sources are geographically
distributed, and the power amount is random and may not enough to power a base station.
For example, solar converters can deliver power only during sunshine hours.
During the night or on cloudy days, other energy sources
have to cover the power demand instead of solar power.
Besides, as large solar
panels are super expensive and take up considerable space, middle size solar powers are always adopted and geographically distributed.

\subsection{ Contributions and Organization of the Work}
Microgrid is the key enabler for deep integration of renewable energy sources.
It is  intelligent, reliable, and environmentally friendly.
In this paper, we consider the case that
base stations are empowered by microgrid.
Brown energy is a stable power source injected in the micogrid when renewable energy is not enough. We investigate how to perform green energy allocation
to abate the dependence on brown energy without sacrificing communication quality.
As geographically distributed green energy need to be delivered to base stations to meet the power demand, power flow needs to be performed
considering power transmission's physical laws constraints i.e., Ohm's law and Kirchhoff's circuit laws.

To solve above problems, in this paper we have made the following major contributions.
\begin{itemize}
\item
In this paper, we systematically study the hybrid energy sources (brown and green energy) powered cellular network.
To evaluate the  greenness of a cellular network, we define the ratio of spectrum efficiency over the total \ce{CO2} consumed (unit:bits/ton\ce{CO2}/Hz) as evaluation metric, and the tradeoff between brown energy consumption and information data rate is established.
\item
The problem of optimal green power allocation to different base stations is modeled to minimize the consumption of brown energy needed
As the power allocation is achieved via power flow over microgrid, power networks physical limits is considered.
More importantly, electricity transmission's physical laws constraints i.e., Ohm's law and Kirchhoff's circuit laws are considered.
\item
Green energy is  time varying in nature. 
Besides, due to measurement error, the exact amount of available green energy is hard to know.
Stochastic online green energy power allocation algorithm is proposed and it is analytically shown that the expected  brown energy consumed converge to the minimum brown energy needed.
\end{itemize}

The rest of the paper is organized as follows.
The network topology of cellular networks empowered by  microgrid with geographically distributed green energy is shown in Section~\ref{Example}.
The brown energy consumption minimization problem is modelled and formulated in Section~\ref{model}.
The online stochastic power allocation algorithm is presented in
Section~\ref{algorithm}.
Simulation results of online power allocation for real power networks are illustrated in Section~\ref{sec-simu}.
Finally, conclusions are given in Section~\ref{conclusion}.
Notations throughout this paper are introduced in in Table \ref{symbol}.


\begin{figure}[b]
  \centering
{\epsfig{figure=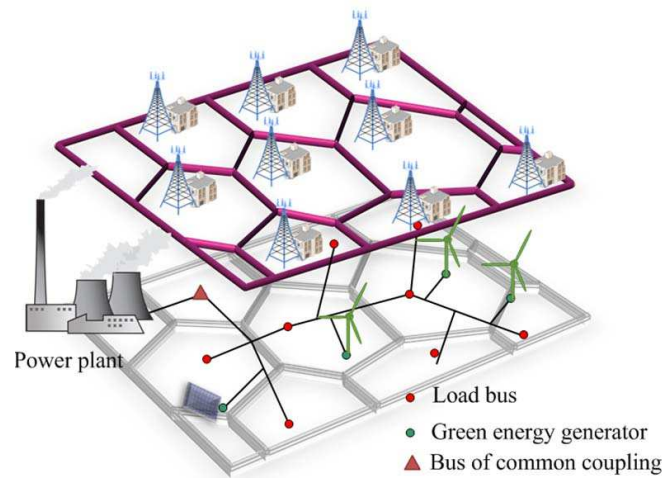, width=3.4in}}
\caption{A green communication networks in the proposed scheme that integrate multiple base
stations powered by  distributed green energy. Brown energy from main grid
is supplied to the communication networks via a bus of common coupling.}
\label{Fignetwork}
\end{figure}

%
\section{System Model and Green Communication Metric}\label{Example}
\begin{figure}[b]
  \centering
{\epsfig{figure=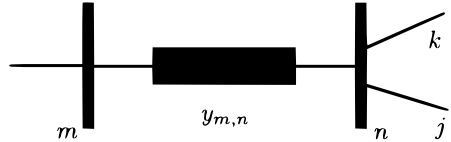,width=2.3in}}
\caption{Equivalent  transmission branch between bus $m$ and bus $n$ with admittance $y_{m,n}$.
$k$ and $j$ are the other two buses connected to bus $n$. }
\label{line}
\end{figure}

\begin{table}
\center
\caption{Notations}
\label{Tab1}
\begin{tabular}{cl|l}
\hline
\rowcolor[HTML]{EFEFEF}
& Symbol & Description\\ \hline
&$N$
& Total number of buses. \\\hline
&$\mathcal{B}$,
& Set of buses connected to base stations. \\\hline
&$\mathcal{G}$,
& Set of buses connected to green energy sources.  \\\hline
& $C_n$
& Communication capacity of base station $n$. \\\hline
& $g_{n}$
& Harvested energy at bus $n$.  \\\hline
& $p_m$
& Power consumption at bus  $m$.
\\\hline
& $E_0$
& Brown energy imported from main grid.
\\\hline
& $v_n$
& Voltage at bus $n$. \\\hline
& $E_0$
& Brown energy injected at bus $0$. \\\hline
& $P_{\textrm{loss}}$
& Total power loss on transmission line. \\\hline
& $\E(\cdot)$
& Expectation over random variable. \\\hline
&$t$       & Index of time slots. \\\hline
&$T$     & Number  of time slots. \\\hline
\hline
\end{tabular} \label{symbol}
\end{table}

The system model of  wireless cellular networks empowered by brown and green energy  is shown in
Fig. \ref{Fignetwork}.
Brown energy is assumed  delivered by the power plant and green energy generated by wind turbines and solar panels are geographically distributed.
Both brown and green energy are transmitted along the microgrid.
The components in this system model are detailed
as follows:

\begin{itemize}
\item
\textbf{Base station:} A base station
 is a centralized device used to
provide wireless communication services to mobile units \cite{AsynNetwork}.
The wireless base station is the power consumption
device.
From the power consumption perspective, the base station has two parts: the communication unit and
the cooling systems.
The communication unit communication with mobile users at certain transmission power and
communication mechanism e.g., 3G, 4G to meet the quality-of-service of mobile users \cite{du2018proactive}.
The cooling subsystem,
including air conditioning and fans, maintains an appropriate
operation temperature at the BS.
Each BS spends more than $50\%$ power on the
cooling system.
The base stations are connected to the micro-grid to obtain stable power supply.

\item
\textbf{Microgrid:}
A microgrid is a localized grouping of electricity/power sources and loads \cite{du2014distributed, wang2016smart}, which
consists of buses and branches, where a bus can represent a generator or a load
substation, and a branch can stand for distribution line.
It connects to the main grid at a bus of common coupling that maintains voltage at the same level
as the main grid,
which provides the microgrid stable power source.
Besides, the increased market penetration of distributed green generation installed, such as
solar, photovoltaic and wind,  has been the advent of an micro-grid.
Voltage difference between different buses pushes current through the branches to supply energy
towards the load bus.
\item
\textbf{Green power source:}
Renewable energy provided from natural resources such as sunlight
and wind is environmentally friendly green energy.
Though clean, no carbon emissions renewable energy is not a panacea.
Renewable energy is less stable and reliable compared to fossil fuel.
Though energy like solar power and wind power is abundant and unlimited, the sun sets or
disappears behind clouds, sometimes the wind stops blowing, causing fluctuations.
Besides, renewable energy exists over wide geographical areas and thus are connected to different
buses.
Thus renewable energy needs to transmit to base stations at different locations.
\end{itemize}

\section{Brown Energy Consumption Minimization: MODELING,
AND FORMULATION
}\label{model}
In the previous section, we have discussed the system architecture for green energy allocation.
In this section, we first provide the metric for our geographically green energy allocation
problem.
We then present
the details for modeling the total brown energy cost, power transmission
constraints, and server service delay constraint as well as
formulating the total brown energy  minimization problem.
Based on detailed discussions, we show that the
constrained optimization problem equivalent to a multiple objective optimization  problem.
The  problem is proved analytically only has Pareto optimal solution.

Considering a general cellular   network consists of $N$ base stations, and  each base stations provides $C_n $ bit/Hz  communication capacity.
Power plants
burns natural gas fuel source
to meet power demand base stations.
In the meantime, green energy, such as solar power and wind power is utilize to abate the dependent on brown energy consumption at the power plant.

Suppose among all power consumption for base stations, $E_0$ kWh power is generated by the brown energy and the carbon dioxide  factor is $\eta$  tons\ce{CO2}/kWh.
Green energy, such as solar power and wind power is utilized to abate the dependent on brown energy.

Though metric like bit/Joule/Hz has been proposed to evaluate the greenness of a communication network \cite{Bhargava}, one can not evaluate how efficiently we utilize green energy.
The emission of \ce{CO2} can be exploited
to distinguish the consumption of brown and green energy.
More specifically, let the $\eta E_0$ be the brown energy injected to the microgrid, the efficiency of utilizing green energy to power the base station networks can  be evaluated by the metric
\begin{equation}\label{metric}
f = \frac{\sum_{n=1}^N C_n}{\eta E_0}.
\end{equation}
The unit of function $f$ is bits/ton\ce{CO2}/Hz.
Next, we model the power flow over micro grid to understand how geographical green energy can be allocated to different base stations.
\subsection{Formulation of Power Flow over Micro-Grid}
The power-flow problem is the computation of voltages and current in a power grid under power line physical constraint.
The power grid consists of buses and branches, where a bus can represent a generator or a load
substation, and a branch stands for a distribution line.
Let $\{0, 1,2,\ldots N\}$ denote the set of bus indexes,
where bus $0$ is the bus of common coupling connected to the main grid directly.
In microgrid, this bus is a slack bus with fixed voltage magnitude denoted by $U$.

Let $\G$ denote the set of $N_g$ green energy sites and $\B$ denote the set of $N_b$ base stations connected to different buses according to their geographical locations.
Two neighboring buses are connected by a transmission line which are represented by a two-port model  with certain admittance as
shown in Fig. \ref{line}.

1) \emph{Brown Energy from Main Grid}:
The amount of power consumption
depends on the type of base station and its communication load.
Let $p_n $  denote the power consumption
by the base station connected to
load bus $n$, and
 $g_m$  denote the energy generation by green energy source connect to generation bus $m$.
In physics, the law of conservation of energy states that the total energy of an isolated system remains constant which is said to be \emph{conserved}.
Therefore, the brown energy needs to be imported from main bus can be expressed as
\begin{equation}
E_0
= \sum_{ n\in \mathcal{B}}p_n
 - \sum_{m\in \mathcal{G}}g_m
 + P_{\textrm{loss}}\footnote{In this paper, we only consider the $E_0>0$ case.},
\end{equation}
where $P_{\textrm{loss}}$ is the total branch losses due to physical heating of the lines by current flow.
For the reason that electricity delived in the micrgrid obeys Kirchhoff's voltage law which states that  the sum of voltages around any closed loop in a
circuit must be zero, the power loss must related to the power generations and consumptions at each bus.
We next investigate the explicit expression of $P_{\textrm{loss}}$.

The admittances of all the branches in the
network are denoted by an admittance matrix $\bm Y$, where the lowercase $y_{mn}$
indicates that matrix element that connects buses $m$ and $n$.
Let $\bm X \triangleq \bm Y^{-1}-\frac{1}{\bm 1^T\bm Y^{-1}\bm 1}\bm Y^{-1}\bm 1\bm 1^T \bm Y^{-1}$, the  power loss $P_{\textrm{loss}}$ of the microgrid is given by \cite{linearPower}
\begin{equation} \label{loss2}
P_{\textrm{loss}}(\{p_n\}_{\forall n\in \B})
= \frac{1}{U^2}\bm p^T\mathcal{R}\{\bm X\}\bm p,
\end{equation}
where
$\bm p =[\bm p^T, -\bm g^T ]^T$ with
$\bm p \triangleq [p_1,...p_{N_b}]^T$ and $\bm g \triangleq [g_1,...g_{N_g}]^T$ \footnote{Only real power is considered in this model and reactive power
 is neutralized by the reactive compensation unit that is installed at base station and renewable energy generator.}.
Due to the fact that elements in the first line and the first column of $\mathcal{R}\{\bm X\}$ are all zeros, $\mathcal{R}\{\bm X\}$ can be denoted by block matrix
$\mathcal{R}\{\bm X\} = \left(\begin{smallmatrix}
0 &   0 & 0\\
0& \bm B & \bm M \\
0 & \bm M^T&\bm G\end{smallmatrix}\right)$, where submatrix $\bm B$, $\bm M$ and $\bm G$ are with compatible dimensions such that (\ref{loss2}) can be equivalently written as
\begin{equation}\label{loss}
 P_{\textrm{loss}}(\{p_n\}_{\forall n\in \B})=\frac{1}{U^2}(\bm g^{T}\bm G\bm g - 2\bm g^T\bm M\bm p + \bm p^T\bm B \bm p).
\end{equation}
From (\ref{loss}), one can notice that the power loss is quadratic function of $p_n$ for all $n\in \B$.

2) \emph{Constraints at Buses and Branches}:
voltage stability played a vital role for the power networks.
For example, in the August 14, 2003, blackout in the
Northeastern U.S.  is due to voltage instability.
In microgrid, voltage regulator which in fact is a transformer tap  is installed for  controlling or supporting voltage.
Denote the voltage at bus $n$ is $v_n$, then it is constrained in the safe range as
$|v_n^{\textrm{min}}|\leq |v_n| \leq v_n^{\textrm{max}} $.
Let $\bm v\triangleq[v_1,\ldots, v_N]^T$
 from (\ref{loss2}) with the  Ohm's law, from (\ref{loss2}) one can  obtain \cite{linearPower}
\begin{equation} \label{voltage2}
|\bm v(\{p_n\}_{\forall n\in \B})|=
 U\bm 1_N  + \frac{1}{U }
[(\bm G \bm g-\bm M\bm p)^T,(\bm N^T\bm g-\bm B\bm p)^T]^T,
\end{equation}
where $\bm 1_N$ is an all $1$ vector with length $N$.

The other constraint need to be considered for microgrid is that base station power consumption at each bus should not exceed the maximum value.
Mathematically, we can denote that
$|p_n^{\textrm{min}}|\leq p_n \leq p_n^{\textrm{max}} $ for $n=1,...N_b$.

Therefore, we formulate the brown energy minimization problem as
\begin{subequations} \label{brown}
\begin{align}
\label{power-1}
\underset{\{p_n\}_{\forall n\in \B}}{\mathrm{min}}
\quad &
\sum_{\forall n\in \mathcal{B}}p_n
 - \sum_{m\in \mathcal{G}}g_m
 +  P_{\textrm{loss}}(\{p_n\}_{\forall n\in \B}) ,\\
\label{Capcaity-2}
\mathrm{s.t.}
\quad &
|p_n^{\textrm{min}}|\leq p_{n} \leq p_n^{\textrm{min}}, \quad \quad \quad \quad \quad \quad  \forall n\in \B \\
\label{subeq4-3}
&
|v_n^{\textrm{min}}|\leq |v_n(\{p_n\}_{\forall n\in \B})| \leq v_n^{\textrm{min}} \quad \forall n\in \B\cup \G.
\end{align}
\end{subequations}

It can be observed from (\ref{brown}),
the objective function is in quadratic form \cite{duJMLR,asilomar} and the constrain domain is convex.

\subsection{Communication Model}
An important metric for characterizing any communication  is the communication capacity which is defined as the maximum amount of
information that can be transmitted as a function
of available bandwidth given a constraint
on transmitted power channel \cite{cai2010cfo}.
For the $n^{\textrm{th}}$ base station, the total energy consumption
$p_n$ includes the
transmitted power  and the rest that is due to other
components such as air conditioning, data processor, and circuits, which can be generally modeled
as a
constant power $p_{c,n} > 0$.
Thus the transmitted power can be denoted by $p_n - p_{c,n}$.
 \cite{Traffic-driven}.

Pertaining to topology constraints \cite{informationmatrix, pairwise} between base station and users, communications can be categorized into uplink and downlink.
An uplink  is  communications used for the transmission of signals from an user to the base station.
As base station only performs received signal processing  for the uplink communication, the corresponding power consumption is assumed incorporated in $p_{c,n}$.
Thus $p_{n}$ is mainly used for downlink transmission.

Next, we investigate the power consumption for downland  transmission.
Without loss of generality, we consider a flat fading model with composite channel fading (including both large-scale and small-scale fading)  from base station $n$ to user $k$ denoted by $h_{n,k}$.
Let $x_{n,k}$ denote the data symbol transmitted from base station $n$, then the received signal $y_{n,k}$ is expressed as
\begin{equation}
y_{n,k} =  (p_n - p_{c,n})x_{n,k}h_{n,k} + w_{n,k},
\end{equation}
where the transmitted power to user $k$ is $\tilde{p}_{n,k}$ joules/symbol,
$h_{n,k}$ denotes the channel fading between base station $n$ and user $k$ and
and $w\sim \N(0,\sigma_n^2)$ is i.i.d. Gaussian noise.

The capacity $C_n $ of a channel provides the performance limit: reliable communication can be attained at any rate $R_k<C_n$;
reliable communication is impossible at rates $R_k>C_n$.
In the multiuser case, the consent capacity region which is the set of all $R_k$ such that all users can reliably communicate at theire own rate $R_k$, respectively.
We have the multi-user data rate bounds \cite{Tse}
\begin{equation} \label{power-capacity}
C_{n}
=\log_2
\big(1 + \frac{(p_n - p_{c,n})|h_{n}|^2}{\sigma^{2}_{n}}\big),\\
\end{equation}
where $|h_{n}| = \arg\min_{h_{n,k}}|h_{n,k}|^2$.
(\ref{power-capacity}) indicates  that increasing the transmission power
$p_n - p_{c,n}$ can increase the transmission capacity.
To achieve the lowest rate, the transmission power is requested to be
\begin{equation} \label{power-capacity}
p_n \geq \underbrace{\frac{(2^{{C}_{n}}-1)\sigma^{2}_{n}}{|h_{n,k}|^2}}_{\underline{p}_n}
+  p_{c,n}.
\end{equation}
Note that information of $\sigma^{2}_{n}$ and $|h_{n}|$ is estimated at user $k$ and feeded back to base station $n$ via control channel.
This feedback communication is a standard process in nowadays wireless communication systems for adaptive power control.
We then have
\begin{subequations}
\label{capacity-network}
\begin{align}
\label{Capcaity-1}
\underset{\{p_n\}_{\forall n\in \B}}{\mathrm{max}}
\quad &
\sum_{n\in \B} C_n,\\
\label{Capcaity-2}
\mathrm{s.t.}
\quad &
p_{n}\geq
\underline{p}_n + p_{c,n},
\quad\quad\quad\quad\quad\quad\quad \forall n\in \B \\
\label{subeq4-3}
&
|p_n^{\textrm{min}}|\leq p_{n} \leq p_n^{\textrm{min}}, \quad\quad\quad\quad\quad\quad \forall n\in \B \\
\label{subeq4-3}
&
|v_n^{\textrm{min}}|\leq |v_n(\{p_n\}_{\forall n\in \B})| \leq v_n^{\textrm{min}},
\quad \forall n\in \B\cup \G
\end{align}
\end{subequations}
It can be simply approved that the above communication capacity optimization is a convex optimization problem.
\subsection{Green Metric for Communication Networks and Pareto Optimal Solvation}
Substituting the geographical power allocation model in (\ref{brown}) and the communication capacity model in (\ref{capacity-network}) in the greenness metric (\ref{metric}) we have
\begin{subequations}
\label{CO2-network}
\begin{align}
\label{CO2-1}
\underset{\{p_n\}_{\forall n\in \B}}{\mathrm{max}}
\quad &
\eta = \frac{\sum_{n=1}^N C_n}
{\eta(\sum_{n\in \B}p_n-\sum_{m\in \G}g_m
+ P_{\textrm{loss}}(\{p_n\}_{\forall n\in \B}) )},\\
\label{CO2-2}
\mathrm{s.t.}
\quad &
p_{n}\geq
\underline{p}_n + p_{c,n},
\quad\quad\quad\quad\quad\quad\quad \forall n\in \B \\
\label{subeq4-3}
&
|p_n^{\textrm{min}}|\leq p_{n} \leq p_n^{\textrm{min}},
\quad\quad\quad\quad\quad\quad \forall n\in \B \\
\label{subeq4-3}
&
|v_n^{\textrm{min}}|\leq |v_n(\{p_n\}_{\forall n\in \B})| \leq v_n^{\textrm{min}},
\quad \forall n\in \B\cup \G
\end{align}
\end{subequations}

The above problem is equivalent to multicriteria optimization problem
and can be formulated as follows
\begin{subequations}
\label{Muti-Obj}
\begin{align}
\label{Muti-1}
\textrm{Obj 1}: \underset{\{p_n\}_{\forall n\in \B}}{\mathrm{max}}
\quad &
C = \sum_{n=1}^N
\log_2
\big(1 + \frac{p_{n}|h_{n}|^2}{\sigma^{2}_{n}}\big),\\
\label{Muti-2}
\textrm{Obj 2}: \underset{\{p_n\}_{\forall n\in \B}}{\mathrm{min}}
\quad &
E_0 = \sum_{n\in \B}p_n-\sum_{m\in \G}g_m+ P_{\textrm{loss}}(\{p_n\}_{\forall n\in \B}) ,\\
\label{Muti-3}
\mathrm{s.t.}
\quad &
p_{n}\geq
\underline{p}_n + p_{c,n},
\quad\quad \forall n\in \B \\
\label{subeq4-3}
&
|p_n^{\textrm{min}}|\leq p_{n} \leq p_n^{\textrm{min}}, \quad \forall n\in \B \\
\label{subeq4-3}
&
|v_n^{\textrm{min}}|\leq |v_n| \leq v_n^{\textrm{min}}.
\quad \forall n\in \B\cup \G
\end{align}
\end{subequations}
It can be easily noticed that these two objective functions are
conflicting objectives.
For example, when the total brown power is
minimized (i.e., $\bm p_n =0$), the communication capacity is also $0$ and is not maximized.
So there does not appear to exist an optimal solution in our problem that optimizes both objectives simultaneously.
Then we can only obtain the Pareto optimal solution when
investigating the multiple objectives optimization problem.

Instead of solving (\ref{Muti-Obj}) directly, let  consider a simpler
single objective optimization problem for a given $C_0$,
\begin{subequations}
\label{Muti-Obj2}
\begin{align}
\label{Muti-2}
\underset{\{p_n\}_{\forall n\in \B}}{\mathrm{min}}
\quad &
E_0(\{{p}_n\}_{\forall n\in \B}) \\\nonumber
=& \sum_{n\in \B}p_n-\sum_{m\in \G}g_m+ P_{\textrm{loss}}(\{p_n\}_{\forall n\in \B}) ,\\
\label{Muti-3}
\mathrm{s.t.}
\quad
& \sum_{n=1}^N
\log_2
\big(1 + \frac{p_{n}|h_{n}|^2}{\sigma^{2}_{n}}\big)\geq C_0,\\
&
p_{n}\geq
\underline{p}_n + p_{c,n},
\quad\quad\quad\quad\quad\quad\quad \forall n\in \B \\
\label{subeq4-3}
&
|p_n^{\textrm{min}}|\leq p_{n} \leq p_n^{\textrm{min}},
\quad\quad\quad\quad \quad\quad \forall n\in \B \\
\label{subeq4-3}
&
|v_n^{\textrm{min}}|\leq
|v_n(\{p_n\}_{\forall n\in \B})| \leq v_n^{\textrm{min}},
\quad \forall n\in \B\cup \G
\end{align}
\end{subequations}
which is named as \textbf{one shot optimization}.
We  show  in the Appendix that
if $E_0=
E_0^{\ast} $ is an optimal solution  for a given value $C_0 = C^{\ast}$ in (\ref{Muti-Obj2}), then $(C^{\ast}, E_0^{\ast})$ is a Pareto optimal solution to (\ref{Muti-Obj}).

\begin{mylemma}
Let $E_0^{\ast}$ be an optimal solution to (\ref{Muti-Obj2}) for a
given value of $C_0=C^{\ast}$, then ($E_0^{\ast}$, $C^{\ast}$) is a Pareto optimal solution to (\ref{Muti-Obj}).
\end{mylemma}
\begin{proof}
See Appendix.
\end{proof}

It can be easily shown  that (\ref{Muti-Obj2}) is a convex optimization problem and therefore can be solved efficiently and globally using the interior-point methods.
However, in practice the power of green energy $g_n$ is random and the measurements of $g_n$ exists errors and lead to error in power allocation.
In the next section, a convergence guaranteed  stochastic renewable energy allocation will be discussed.

\section{Online Stochastic  Power Allocation  }\label{algorithm}
\subsection{Average Brown Energy Consumption Optimization }
In microgrid, Algorithm 1 is implemented in a central control unit for power allocation via gathering the green energy $g_m$ and other parameters in the micrograms and cellular networks.
However, real-time exact value of $g_m$ is difficult to be obtained
due to measurements errors and communication delays from distributed green energy sites to the central control unit.
Rather than implementing the unreliable and possibly obsolete instantaneous decision of $p_n$ of (\ref{Muti-Obj2}), we consider
a robust  stochastic control scheme which is to minimize the average brown energy consumption as
\begin{subequations}
\label{stochastic}
\begin{align}
\label{Muti-2}
&\underset{\{p_n(t)\}_{\forall n\in \B}}{\mathrm{min}} f(\{g_n(t)\}_{\forall n\in \B})
\\\nonumber
&=
\E\big\{ \sum_{n\in \B}p_n(t)-\sum_{m\in \G}g_m(t)
+ P_{\textrm{loss}}(\{p_n(t)\}_{\forall n\in \B}) \big\} ,\\
\label{Muti-3}
\mathrm{s.t.}
\quad
& \sum_{n=1}^N
\log_2
\big(1 + \frac{p_{n}(t)|h_{n}|^2}{\sigma^{2}_{n}}\big)\geq C_0,\\
&
p_{n}\geq
\underline{p}_n(t) + p_{c,n}(t),
\quad\quad\quad\quad \forall n\in \B \\
\label{subeq4-3}
&
|p_n^{\textrm{min}}|\leq p_{n}(t) \leq p_n^{\textrm{min}}(t),
\quad\quad\quad \forall n\in \B \\
\label{subeq4-3}
&
|v_n^{\textrm{min}}|\leq
|v_n(\{p_n(t)\}_{\forall n\in \B})| \leq v_n^{\textrm{min}},
\forall n\in \B\cup \G
\end{align}
\end{subequations}
In the subsequent, we use $p_n(t)$ and $g_m(t)$ to represent the time varying  power consumption and green power generation of base station $n$ and green energy site $m$.
\subsection{Stochastic Online Power Allocation Algorithm Design and Analysis}
Leveraging recent advances in online stochastic convex optimization, the above problem (\ref{stochastic}) can be solved under stochastic approximation framework.
In this paper, we use Bregmen projection based mirror decent algorithm to iterate the objective
variable $\bm p(t)$ that gradually converges to the optimization point of the expectation
functions in (\ref{stochastic}).

Let $f^{\ast}(\bm y)\triangleq \sup_{\bm x\in \textrm{dom}f}(\bm x^T\bm y -f(\bm x))$ denote the
conjugate function of $f(\bm x)$,
and $F(\cdot)$ denotes a continuous differentiable function that is $\alpha$-strongly
convex w.r.t. Euclid norm $\Vert \cdot \Vert$.
The Bregmen divergence associated with $F(\cdot)$ is
defined as \cite{OnlineOpt}
\begin{equation}\label{Bregmen}
B_{F}(\bm x, \bm y)
\triangleq
F(\bm x)-F(\bm y)- (\bm x - \bm y)^T  \triangledown F(\bm y) ,
\end{equation}
where $B_{F}$ satisfies  $B_{F}(\bm x, \bm y)\geq \frac{\alpha}{2}\Vert  {\bm x-\bm y} \Vert^{2}$ for some $\alpha>0$.
The online mirror decent method is described as the following two steps iterations:
First, the gradient is performed in the dual space rather than in the primal space as
\begin{equation}
\label{iteration1}
\omega_{\bm p(t+1)} =\triangledown F^{*}(\triangledown F(\bm p(t))-\delta \triangledown
f(\bm p(t), \bm p_g)),
\end{equation}
where $F^{*}(\cdot)$ denotes the dual function of $F(\cdot)$.
and $\delta $ is the step size.
The second step is the projection step defined by the Bregman divergence associated to $F(\cdot)$ as
\begin{equation}
\label{iteration2}
\bm{p}(t+1) =\arg \min_{\bm x\in \mathcal{A}} B_{F}(\bm x,\omega_{\bm p}(t+1)),
\end{equation}
where $\mathcal{A} $ is the feasible domain of $\bm x$.
Intuitively, the above mirror decent algorithm minimizes a first-order approximation of the function $f(\cdot)$ at the current iterate $\omega_{\bm p(t)}$ while forcing the next iterate $\omega_{\bm p(t+1)}$
to lie close to $\omega_{\bm p(t)}$.
The step size $\delta$ controls the trade-off between these two.

To obtain closed form solution in the iteration, we use standard norm as the Bregmen divergence function i.e., $F(\bm x)=\frac{1}{2}\Vert\bm x
\Vert_{2}^{2}$.
Substituting the expression into  (\ref{Bregmen}) and after simple calculation,  we obtain
$B_{F}(\bm x, \bm y)= \frac{1}{2}\Vert \bm  {x-y} \Vert^{2}$.
By Holder's inequality, and a simple optimization of a quadratic polynomial, one has
$F^{*}(\bm y) = \sup_{x\in \mathcal{D}}(\bm x^T\bm y - \frac{1}{2}\Vert\bm x\Vert^2)\leq
\sup_{x\in \mathcal{D}}
(\Vert\bm x \Vert \Vert\bm y\Vert-\frac{1}{2}\Vert\bm x \Vert^2)=\frac{1}{2}
\Vert  {\bm x-\bm y} \Vert^{2}$.
Notice that the inequality above is in fact an equality by definition of the dual norm.
Thus $F^{*}(\bm y) = \bm y$.
Moreover, as
$\triangledown F(\bm x) = \bm x $,
we have
$\triangledown F(\bm x)^{\ast} = \bm x $.

Then the update with the gradient in (\ref{iteration1}) can be computed as
\begin{equation}
\label{iteration1-1}
\omega_{\bm p(t+1)} =\bm p(t)-\eta \triangledown f(\bm p(t), \bm g).
\end{equation}
where
\begin{equation}\label{gre}
\triangledown f(\bm p(t), \bm g) =
-2\bm g^T\bm M + 2\bm p(t)^T\bm B.
\end{equation}
Then the update with the projection step in (\ref{iteration2}) is derived as
\begin{subequations}
\label{iteration2-1}
\begin{align}
\label{Muti-2}
\bm{p}(t+1) =& \arg \min \frac{1}{2}\Vert \bm x - \omega_{\bm
p}(t+1)\Vert^2 ,\\
\label{Muti-3}
\mathrm{s.t.}
\quad
& \sum_{n=1}^N
\log_2
\big(1 + \frac{p_{n}|h_{n}|^2}{\sigma^{2}_{n}}\big)\geq C_0,\\
&
p_{n}\geq
\underline{p}_n + p_{c,n},
\quad\quad  \forall n\in \B \\
\label{subeq4-3}
&
|p_n^{\textrm{min}}|\leq p_{n} \leq p_n^{\textrm{min}},
\quad\quad \forall n\in \B \\
\label{subeq4-3}
&
|v_n^{\textrm{min}}|\leq
|v_n(\{p_n\}_{\forall n\in \B})|
 \leq v_n^{\textrm{min}}.
\quad \forall n\in \B\cup \G
\end{align}
\end{subequations}

To summarize, the stochastic online power allocation can be solved by an iterative algorithm with two steps in each iteration.
The first step can be easily computed using the closed form of the
gradient as in (\ref{iteration1-1})
The second step is to compute (\ref{iteration2-1}).
Obviously, (\ref{iteration2-1}) is a simple convex optimization problem,
 which can be solved by interior point method.
The stochastic online power allocation  algorithm is summarized in Algorithm 1.

Notice that the proposed online power allocation algorithm does not depend on any measurement error or communication delay distribution assumption. It rather utilizes real-time communication and power data to infer the unknown statistics.
Next we evaluate the performance of the online power allocation algorithm.
\begin{algorithm}[t]
\caption{Stochastic Online Power Allocation}
\begin{algorithmic}[1]
\STATE  {Initialize: }
  {$\bm p(1)$ is computed by one shot optimization in (\ref{Muti-Obj2})}.
\FOR{  $t \in\{1,2^, \ldots, T\}$
}
    \STATE  Compute the gradient $\omega_{\bm p(t+1)}$ according to (\ref{iteration1-1}).
    \STATE Update the gradient according to (\ref{gre}).

\STATE    Compute $p(t+1)$ via minimizing $B_{F}(\bm x,\omega_{\bm p}(t+1))$  according to (\ref{iteration2-1}).
\ENDFOR
\end{algorithmic}
\end{algorithm}

\begin{myproperty} \label{P-Cov}
Let $\{\tilde{p}_n(t)\}_{\forall n\in \B}$ and $\{\hat{p}_n(t)\}_{\forall n\in \B}$ are the minimizers of (\ref{Muti-Obj2}) and (\ref{stochastic}) respectively, it holds that \begin{equation}
E_0(\{{\tilde{p}}_n(t)\}_{\forall n\in \B}) -
f(\{\hat{p}_n(t)\}_{\forall n\in \B})
\leq \frac{k}{\sqrt{T}},
\end{equation}
where $k$ is a constant irrespective of $T$.
 \end{myproperty}

Proposition 1 guarantees that the expected brown power consumption converges to the optimum stochastic solution at the rate of $\mathcal{O}(1/\sqrt{T})$.
Therefore the online power allocation algorithm
in Table 1 has sublinear accurate error.

\section{PERFORMANCE EVALUATION}\label{sec-simu}
In this section, we present and discuss simulation results, in
comparison with numerical results pertaining to the previously
developed analysis.

\subsection{Parameter Setting}
The $37$ bus test feeder model \cite{IEEE_PES_Test_Feeders} is used for the simulation of micogrid.
It represents an actual underground radial distribution feeder in California.
The schematic view of this network can
be seen in Fig. \ref{1234}.
In this $37$-bus network,
bus $799$ is the bus of common coupling.
Stable brown energy is injected to the microgrid when green energy is not enough.
Green energy sites are geographically distributed at buses $742$, $725$, $735$ $731$ and $741$.
Each other buses is load bus and is linked to a base station.
The maximum power rate at each bus is set as $1.3$MW.
Solar power generation
sites of $12$MW are placed respectively on bus $742$, $725$ and $735$ and the wind power generation sites of $8$MW are placed respectively on bus $731$ and $741$.
The green energy data set are from Smart${^{\ast}}$ Microgrid Data Set \cite{dataset}.
Besides, we generate the data set for base station power consumption
by adopting Huawei DBS3900 base station power consumption parameters in the simultion \cite{Huawei}.
User traffic data are loaded according to \cite{userpattern} for the simulation and the total user number variation is given in Fig. \ref{no}.
In general while users increases, the power consumption increases.

For each operation interval, the microgrid control center collects green energy generation $g_m $ from generation  buses.
It is assumed that the renewable generation data is observed with a white Gaussian
noise of $20\%$ of its actual value due to measurement error.
Besides, the communication channel state information
$h_n $ for all $n\in \B$ and noise variance $\sigma_n^2$ are obtained in the communication networks and send to the microgrid control center for implement the proposed online stochastic power allocation algorithm.

\subsection{Experimental Results}

While the user traffic increases more power is needed to support base station for information transmission.
As the transmission power increases, the communication capacity $\sum_{n\in\B}C_n$ will increase.
Assume  frequency reuse is adopted and therefore there is no  inter-cell interference beween neighboring cells.
To investigate the advantage of green energy  for the communication system, base stations communication power at each base station is normalized by $\frac{1}{N_b}(P_{\textrm{loss}}+ E_0)$.
Fig. 5 and Fig. 6
show the capacities variation of base stations at bus $720$ and $730$, respectivly.
It is clear the proposed online stochastic power allocation algorithm provides larger communication capacity then the one shot solution.
Thus it  utilizes the green energy more efficiently and therefore abates the dependence on brown energy from the main grid.
Besides, it can be observed that online stochastic algorithm provides more stable capacity then one shot algorithm.

Power loss on the microgrid is   due to  physical heating  on the transmission branches caused by electricity delivering.
Though it is unavoidable, it is preferred to lower the power loss as much as possible.
Distributed green energy also decrease the power loss on the microgrid.
We evaluate the total power loss for
the proposed stochastic online power allocation algorithm and one shot optimization algorithm.
Besides,
the power loss is also simulated according to (\ref{loss})
 in the case that there is no green energy distributed in the microgrid and all  energy is brown energy imported from the main grid.
Fig. 7 shows the total power loss cost
reductions by comparing results of the proposed online stochastic power allocation algorithm (Algorithm 1)
and the one shot solution  (\ref{Muti-Obj2}).
It shows that while geographically green energy incorporated, powr flow is more efficiently as the power source is not far from the consumption bus.
In contrast, with only brown energy from main grid, the power loss increases greatly.

\begin{figure}[t]
  \centering
{\epsfig{figure=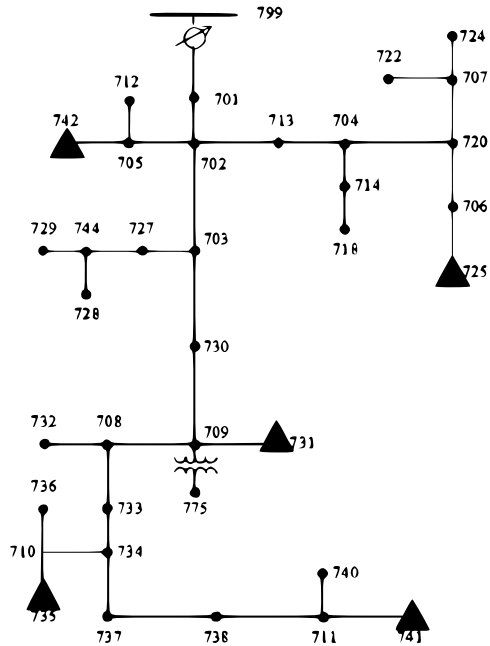,width=2.5in}}  \caption{Schematic diagram of the IEEE 37-bus test feeder cases \cite{IEEE_PES_Test_Feeders}.
Bus $799$ is the bus of common coupling that links the microgrid to main power grid.
Solar power harvesting sites are located at buses $742$, $725$, $735$ and wind power harvesting sites are located at buses  $731$ and $741$. These green energy sites are denoted by green triangles.
Other buses  are load buses connected to base stations.
}
\label{37test}
\end{figure}

\begin{figure}[t]
  \centering
{\epsfig{figure=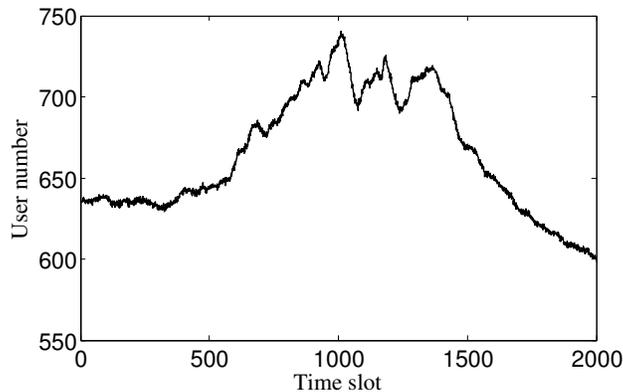,width=3.2in}}
\caption{Dynamic user number versus time slot. }
\label{no}
\end{figure}

\begin{figure}[t]
  \centering
{\epsfig{figure=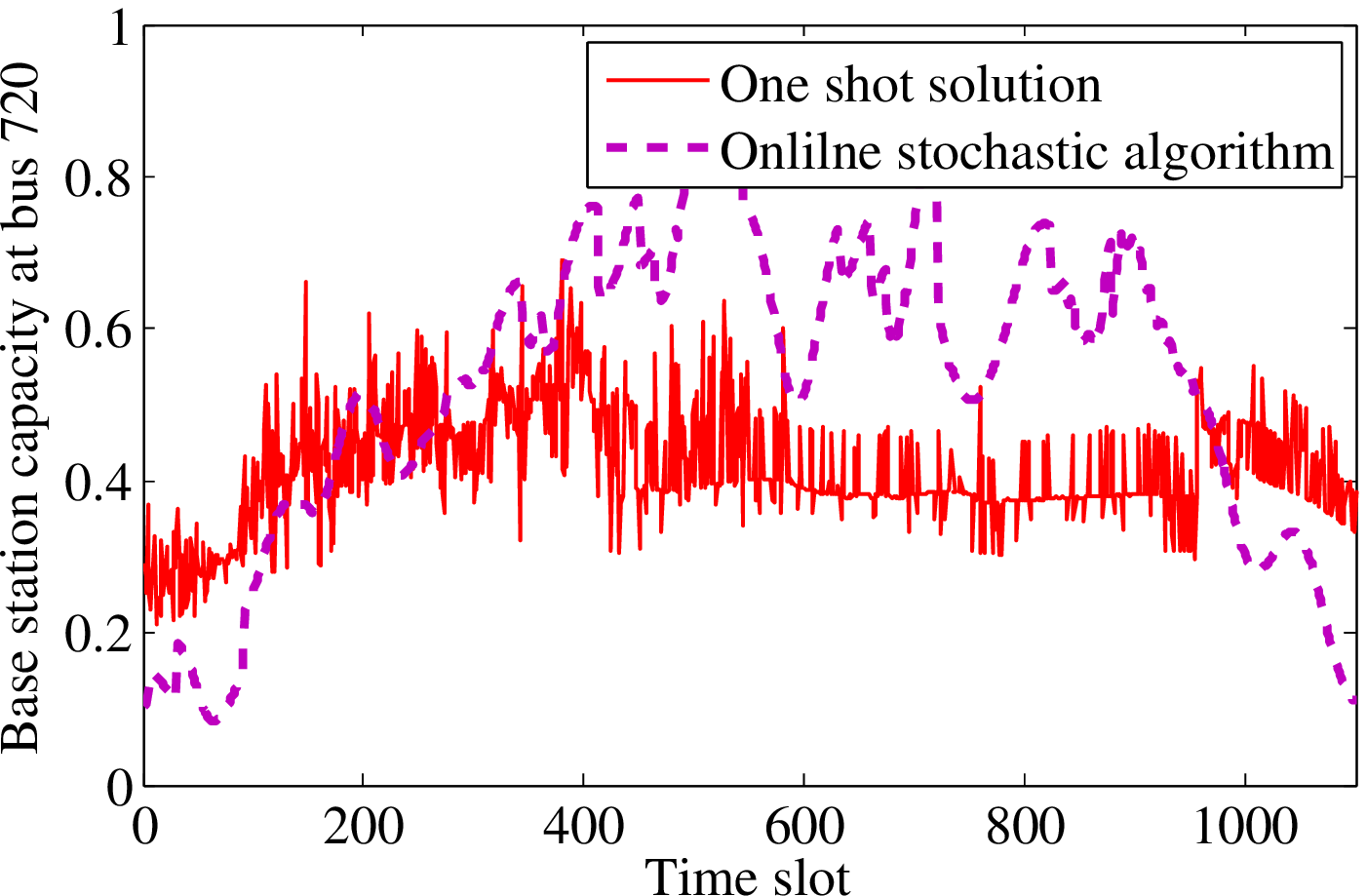,width=3.0in}}
\caption{Normalized communication capacities of base station empowered by bus $720$ versus time slot.}
\label{MSE-Iter-PDR}
\end{figure}

\begin{figure}[t]
  \centering
{\epsfig{figure=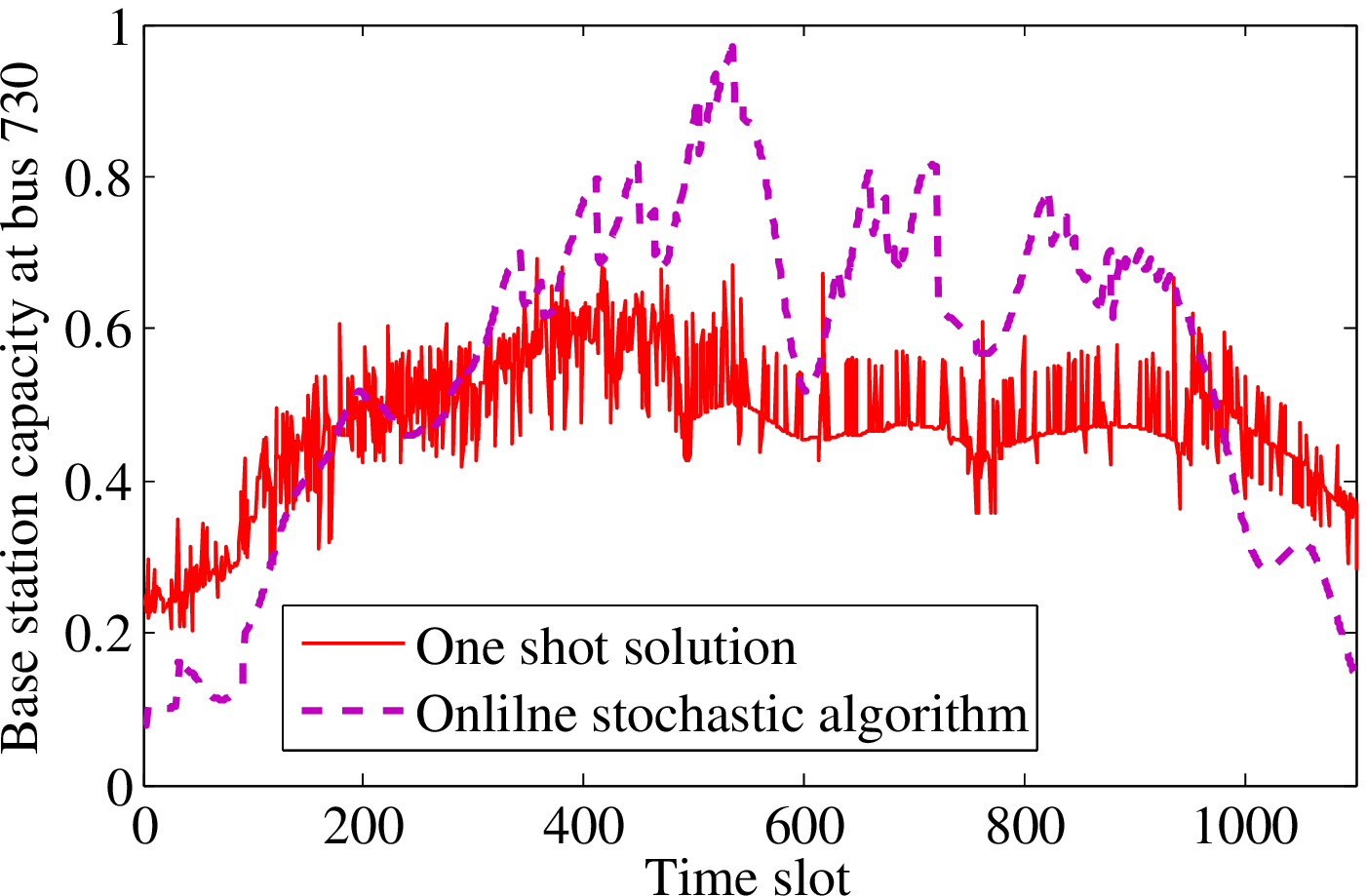,width=3.0in}}
\caption{Normalized communication capacities of base station empowered by bus $730$ versus time slot. }
\label{MSE-dynamic}
\end{figure}

\begin{figure}[t]
  \centering
{\epsfig{figure=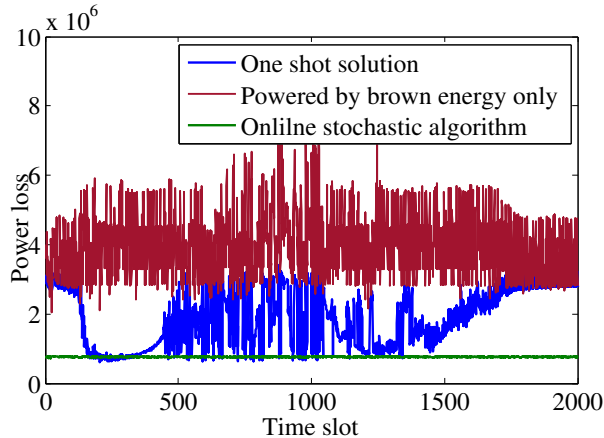,width=3.1in}}
\caption{Power loss versus time slot.  }
\label{1234}
\end{figure}

\section{Conclusions}\label{conclusion}
In this paper, we have investigated optimal green energy allocation to abate the dependence on brown energy without sacrificing communication quality.
Leveraging the fact that the green energy power amount generated from distributed geographical locations varies, the framework of green energy cooperation to empower cellular networks has been achieved  by formulate the problem as  maximizing data rate per ton \ce{CO2} consumed.
Electricity transmission's physical laws constraints i.e., Ohm's law and Kirchhoff's circuit laws, are considered for the power allocation.
The exact real-time amount of green energy is difficult to obtain due to
the noise contaminated measurement and possible communication delay from energy harvesting site to control unite, and therefore
degrades the power allocation performance.
We therefore have proposed stochastic online power allocation algorithm  which is robust to the green energy uncertainty.
It is shown that the online algorithm converges to the optimum stochastic data rate per ton \ce{CO2}.
Finally, we have conducted  data trace-driven
experiments.
The results show that
geographically distributed green energy complement each other by improving the communication capacity while saving brown energy consumption from main grid.
We also come to some key findings
such as with penetration of green energy, power loss on the transmission breaches can be reduced.

\appendix
The proof is based on contradiction. Suppose ($E_0^{\ast}$,
$C^{\ast})$ is not a Pareto optimal solution to (\ref{Muti-Obj}). Then there must
exist a feasible solution ($E_0$, $C$) to (\ref{Muti-Obj}) such that $C > C^{\ast}$
and $E_0 < E_0^{\ast}$. Based on this feasible solution ($E_0$,  $C$), we can
construct another feasible solution to (\ref{Muti-Obj2}) as follows. We
keep power consumption unchanged, but decrease $C$ to $C^{\ast}$
through scheduling or interference management. Clearly, ($ E_0,
C^{\ast}$) is also a feasible solution. So corresponding to the same
$C^{\ast}$ value, we have two feasible solution ($E_0, C^{\ast}$) and ($E_0^{\ast},
C^{\ast}$) and that $E_0 < E_0^{\ast}$. This contradicts to the fact that $E_0^{\ast}$ is
an optimal solution to (\ref{Muti-Obj2}) under a given $C^{\ast}$ value. This
completes our proof.


\end{document}